\documentclass{llncs}

\usepackage[utf8]{inputenc}
\usepackage[T1]{fontenc}
\usepackage[english]{babel}

\usepackage{amsmath,amssymb}
\usepackage{graphicx}
\usepackage{color}
\usepackage{url,ifthen}
\usepackage{fancyhdr}
\usepackage{amsmath} 
\usepackage[pagewise]{lineno}

\usepackage[dvipsnames]{xcolor}
\usepackage{booktabs}
\usepackage[ruled,vlined]{algorithm2e}
\newboolean{anonymous}
\setboolean{anonymous}{false}

\newcommand{\C}{\mathcal{C}}

\newcommand{\F}{\mathbb{F}}

\begin{document}
	\title{Twisted Conjugacy and the Construction of Induced Centrosymmetric Alternant Codes}

\ifthenelse{\boolean{anonymous}}
{ 
\author{
    Anonymized for submission
}
\institute{
}
} 
{ 
	\author{
   Ousmane Ndiaye \and
Massamba  Sow
}
\institute{
Universit\'e Cheikh Anta Diop de Dakar, FST, DMI, LACGAA,\\ Senegal,\\
\email{ousmane3.ndiaye@ucad.edu.sn}\\
\email{massamba2.sow@ucad.edu.sn}
}

} 
	\date{\today}
	\maketitle
\begin{abstract}
This paper presents a construction of induced centrosymmetric alternant codes based on the automorphism structures inherited from Generalized Reed–Solomon (GRS) codes. We introduce a twisted conjugacy relation naturally arising from projective semilinear transformations and establish its connection with ordinary conjugacy in the corresponding projective semilinear coset. By means of the Shintani norm, we show that, in the quadratic extension case, all admissible transformations inducing centrosymmetry have a trivial Shintani norm and therefore belong to a unique conjugacy class in the semilinear coset. We further characterize the explicit forms of these transformations and establish a transport theorem describing how twisted conjugation acts on the supports and multipliers of GRS codes. These results provide an algebraic framework for constructing and analyzing induced centrosymmetric alternant codes.
\par\textbf{MSC classes:}	20H30, 94B05, 20G40
\end{abstract}

\section{Introduction}\label{}
Error-correcting codes play a central role in information theory by ensuring reliable data transmission over noisy channels. Among the numerous families of algebraic codes, Generalized Reed--Solomon (GRS) codes occupy a prominent position because of their optimal error-correcting capabilities and rich algebraic structure. Since alternant codes are obtained as subfield subcodes of GRS codes, they inherit many of their structural properties while encompassing several important code families.

The study of automorphism groups of linear codes has attracted considerable attention because these groups capture the algebraic symmetries of the underlying code. Such symmetries provide valuable insight into the structural and combinatorial properties of codes and have led to the construction of several remarkable families, including cyclic, quasi-cyclic \cite{10.1007/978-3-642-02384-2_6} and quasi-dyadic \cite{MB09} codes. Consequently, understanding how automorphisms of GRS codes are inherited by alternant codes is a fundamental problem in coding theory.

Recent work by Ndiaye~\cite{Ousmane2023} introduced a family of
quasi-centrosymmetric alternant codes obtained from
quasi-centrosymmetric Cauchy matrices. These codes arise from the particular involution induced by
the projective transformation $\sigma(x)=x+b$, and admit an
automorphism group of order two. This naturally raises the broader
question of determining and classifying centrosymmetric alternant
codes induced by automorphisms of the underlying generalized
Reed--Solomon codes.

In this paper, we address this question by establishing an automorphism-based classification of induced centrosymmetric alternant codes. Our approach is based on the identification of the twisted conjugation action naturally associated with semilinear projective transformations. This identification establishes a bridge between twisted conjugacy in $PGL_2(\mathbb F_q)$ and ordinary conjugacy in $P\Gamma L_2(\mathbb F_q)$, allowing us to use the Shintani norm \cite{shintani1976two} to classify the induced centrosymmetric alternant codes.
\subsection{Related Works}

The study of automorphisms of algebraic codes has a long history. Stichtenoth \cite{STICHTENOTH1989205} characterized cyclic extended Goppa codes, while Berger \cite{BERGER2000255} showed that parity-check subcodes and extended Goppa codes are alternant codes and introduced new cyclic and non-cyclic constructions. Berger subsequently investigated Goppa and related codes invariant under permutations induced by elements of the projective semilinear group, highlighting the close relationship between automorphisms of GRS codes and the resulting alternant codes.

More recently, Ndiaye~\cite{Ousmane2023} introduced
quasi-centrosymmetric alternant codes associated with the particular
involution
\[
\sigma(x)=x+b,
\]
which satisfies $\sigma^2=\operatorname{id}$ in characteristic two.
This construction naturally raises the following question:
what happens when one considers the full family of projective
semilinear transformations $x\longmapsto\big( \frac{a x+b}{cx+d} \big)^{p^j}~?$

The present work extends this construction by considering all conjugacy classes of involutions in the projective semilinear group $P\Gamma L_2(\mathbb{F}_q)$. This leads to a systematic classification of induced centrosymmetric alternant codes within a common automorphism-based framework.

The work closest to ours is Berger \cite{Berger1999}, who classified the cyclic alternant codes induced by a semilinear automorphism of a GRS code. His method has three steps: a characterization of the induced codes by conditions on the support and on the multiplier \cite[Theorem 2.2]{Berger1999}, the fact that the support can be changed by a homography \cite[Proposition 2.3]{Berger1999}, and a classification of the elements of $P\Gamma L_2(\mathbb{F}_q)$ up to conjugation by their number of fixed points \cite[Theorem 2.4]{Berger1999}. We follow the same steps for the centrosymmetric permutation $i\mapsto n-1-i$ instead of the cyclic shift. The main difference is that this permutation is an involution. The support is then a union of orbits of length at most two, and the square of the semilinear transformation fixes every point of the support. This gives constraints which do not appear in the cyclic case. The construction of \cite{Ousmane2023} corresponds to the case $j=0$ and $p=2$.


\subsection{Our Contributions}
The contributions of this paper are the following.

\begin{itemize}
\item We characterize the induced centrosymmetric alternant codes by conditions on the support and on the multiplier (Theorem~\ref{tho_3}). This is the analogue of the result of Berger for cyclic codes.
\item We give a complete proof of the change of support under a homography, together with the relation satisfied by the new multiplier (Theorem~\ref{lemmegamma}). We express it with a twisted conjugation on \(\operatorname{PGL}_{2}(\mathbb{F}_{q})\), which corresponds to ordinary conjugacy within the semilinear coset of \(\operatorname{P\Gamma L}_{2}(\mathbb{F}_{q})\) (Theorem~\ref{lem:semilinear-conjugacy}).
\item We show that one can always assume $j\in\{0,\frac m2\}$: otherwise the support and the multiplier are in a proper subfield (Proposition~\ref{prop:reduction}).
\item For $j=\frac m2$ we refine the classification of \cite{Berger1999}: we give the parameters of the classes, we prove that the classes are distinct and we count them (Theorem~\ref{classification_grs}). We prove that all the semilinear involutions have trivial Shintani norm and belong to a unique class (Corollary~\ref{cor:involutions}).
\item We obtain normal forms for the support and the multiplier of every induced centrosymmetric alternant code (Theorem~\ref{thm:normal}), and we show that the four usual families give the same codes when $j=\frac m2$.
\item We give a criterion for the existence of a centrosymmetric generator matrix (Proposition~\ref{prop:centro_generator}) and examples over $\mathbb{F}_{64}$.
\end{itemize}

\subsection{Organization of the Paper}
The rest of the paper is organized as follows. In Section~\ref{sec:prelim}, we recall the necessary background on generalized Reed–Solomon codes, alternant codes, with a particular focus on their automorphism groups. In Section~\ref{sec:centro}, we present several types of constructions for induced centrosymmetric alternant codes. An example of the constructions with semilinear transformations is given in Section~\ref{exemple}.  Finally, Section~\ref{sec:conclusion} concludes the paper and discusses future research directions.

\section{Preliminaries} \label{sec:prelim}
Let $q=p^m$, where $p$ is a prime number, $k$ and $n$ positive integers such that $2 \leq k < n \leq q$.\\
\subsection{Automorphism group of a code}
\begin{definition}(\textbf{Permutation group})\\
		Let $\C$ be a linear code of  length $n$ over $\F_q$ and $\sigma \in S_n$ acting on $\F_q^n$ by $\forall c \in \F_q^n, \sigma(c)=(c_{\sigma^{-1}(0)},...,c_{\sigma^{-1}(n-1)})$.\\
		 The code $\C$ is said $\sigma$-invariant if $\sigma(\C)=\C$.\\
		  The permutation group of  $\C$ is:
		 $Perm(\C)=\{ \sigma \in S_n | \sigma(\C)=\C\} $
\end{definition}

More generally, if  $q=p^m$, the group $(\F_q^{*})^n\rtimes ( S_n \times Gal_{\F_p}(\F_q))$ is defined by the law  $(e, \sigma, \gamma) (\epsilon,\tau ,\pi ):= (d, \sigma \circ \tau ,\gamma\circ\pi)$ with $d_i=e_i\gamma(\epsilon_{\sigma^{-1}(i)})$.
This group is called the semilinear isometry group of $\mathbb{F}_q^n$ with respect to the Hamming metric.

\begin{definition}(\textbf{semilinear automorphism group of a linear code})\\
	Let $\C$ be a linear code of  length $n$ on $\F_q$.\\
    Let $\sigma \in S_n$, $\gamma \in Gal_{\F_p}(\F_q)$ and $e \in (\F_q^{*})^n$ acting on $\F_q^n$ by $\forall c \in \F_q^n, (e,\sigma, \gamma)(c)=(\gamma(e_{\sigma^{-1}(0)})\gamma(c_{\sigma^{-1}(0)}),...,\gamma(e_{\sigma^{-1}(n-1)})\gamma(c_{\sigma^{-1}(n-1)}))$. The code $\C$ 
	is said $(e,\sigma, \gamma)$-invariant if $(e,\sigma, \gamma)(\C)=\C$. The semilinear automorphism group of $\C$ is:
	$$\Gamma Aut(\C)=\{ (e,\sigma, \gamma) \in (\F_q^{*})^n\rtimes ( S_n \times Gal_{\F_p}(\F_q)) ~|~ (e,\sigma, \gamma)(\C)=\C\} $$
\end{definition}

\begin{remark}~\\
	\begin{itemize}
		\item We define $Aut(\C)=\{ (e,\sigma) \in (\F_q^{*})^n\rtimes ( S_n) ~|~ (e,\sigma, Id_{\F_q})(\C)=\C\}$ as the linear automorphism group of $\C$ which is the monomial group of square matrix of order $n$ on $\F_q$, $M_n(\F_q^{*})$. If $q$ is prime $Aut(\C)= \Gamma Aut(\C)$ 
		\item $Perm(\C)\subseteq Aut(\C)\subseteq \Gamma Aut(\C)$.
		\item if $q=2$ then $Perm(\C)=Aut(\C)=\Gamma Aut(\C)$
	\end{itemize}

\end{remark}

\begin{remark}~\\

	$ \forall (e,\sigma, \gamma) \in (\F_q^{*})^n\rtimes ( S_n \times Gal_{\F_p}(F_q))$ then
	\begin{itemize}
		\item The map $(e,\sigma,\gamma)$ is a semilinear Hamming isometry, i.e., it preserves the Hamming weight.
		\item $(e,\sigma, \gamma)(\C)$ is a linear code of the same length, the same dimension and the same minimal distance as $\C$. 
	\end{itemize}
\end{remark}
$$\Gamma Aut(\C) \simeq \Gamma Aut(\C^\bot)$$
$$Perm(\C)=Perm(\C^\bot)$$
The isomorphism is $(e,\sigma,\gamma)\mapsto(e^{-1},\sigma,\gamma)$, where $e^{-1}=(e_0^{-1},\dots,e_{n-1}^{-1})$. Indeed, for $a,b\in\F_q^n$ we have $\langle (e,\sigma,\gamma)(a),(e^{-1},\sigma,\gamma)(b)\rangle=\gamma(\langle a,b\rangle)$, so $(e,\sigma,\gamma)(\C)=\C$ if and only if $(e^{-1},\sigma,\gamma)(\C^\bot)=\C^\bot$.
\begin{proposition}
	Let $\sigma \in S_n$, $\C$ a linear code, $B$ a basis of  $\C$. If $\sigma(B)=B$ then  $\sigma \in Perm(\C)$
\end{proposition}
 
\subsection{GRS codes and Cauchy codes} 
Let $\mathbf{x}=(x_0,...,x_{n-1}) \in \mathbb{F}_q^n$ such that $x_i$ distinct in $\mathbb{F}_q$, $\mathbf{y}=(y_0,...,y_{n-1}) \in \mathbb{F}_q^n$ Such that $y_i$ non-zero scalar.
\begin{definition}(\textbf{Generalized Reed-Solomon code})\\
	The Generalized Reed-Solomon code $GRS_k(\mathbf{x},\mathbf{y})$ of dimension $k$ associated to $(\mathbf{x},\mathbf{y})$ is the $\mathbb{F}_q$-vector space of dimension $k$ defined by:
	$$ GRS_k(\mathbf{x},\mathbf{y})=\{(y_0P(x_0),...,y_{n-1}P(x_{n-1})) | P \in \mathbb{F}_q[z]_{<k} \}$$
	$\mathbf{x}$ is called support and $\mathbf{y}$ multiplier of $GRS_k(\mathbf{x},\mathbf{y})$.\\
	The Reed-Solomon code corresponds  to the case where $y_i=1, \forall$ $i \in \{0,..., n-1\}$ and is denoted by $RS_k(\mathbf{x})$.
\end{definition}
To better treat the group of automorphisms of GRS codes A. Dür proposed in \cite{DUR198769} a projective generalization on $\mathbb{P}^1(\mathbb{F}_q)=\mathbb{F}_q\cup \{\infty\}$ of GRS codes which amounts to consider a support with an additional point: infinity. This new code obtained is called the Cauchy code $C_k(\mathbf{x},\mathbf{y})$  defined from homogeneous bivariate polynomials of degree equal to $k-1$.
\begin{proposition}(Generator Matrix)
	Let $\mathbf{x}$ and $\mathbf{y}$ respectively support and multiplier from $\mathbb{F}_q$, then
	$$G=\left(
	\begin{array}{cccc}
		y_0 & y_1 & \cdots  & y_{n-1} \\
		y_0x_0 & y_1x_1 & \cdots& y_{n-1}x_{n-1} \\
		\vdots & \vdots &   & \vdots \\
		y_0x_0^{k-1} & y_1x_1^{k-1} & \cdots  & y_{n-1}x_{n-1}^{k-1} \\
	\end{array}
	\right)
	$$
	is a generator matrix of $C_k(\mathbf{x},\mathbf{y})$.
\end{proposition}
\subsection{Automorphism group of Cauchy code}
A. D\"ur, in \cite{DUR198769}, was able to establish the automorphism group from the semilinear Projective $P\Gamma L_2(\mathbb{F}_q)=PGL_2(\mathbb{F}_q)\rtimes Gal(\mathbb{F}_q)$ group acting on $\mathbb{P}^1(\mathbb{F}_q)$.
GRS codes with symmetries can easily be constructed by looking at the action on the support and the multipliers of $PGL_2(\mathbb{F}_q)$. 
In the following for more simplicity, we will adopt the notations in \cite{Berger1999}.
For all $f=\left(
\begin{array}{cc}
	a & b \\
	c & d \\
\end{array}
\right) \in GL_2(\mathbb{F}_q)$ by projection on $PGL_2(\mathbb{F}_q)$ $f$ induces a homography defined by 

$$\overline{f}(z)=\frac{az+b}{cz+d}$$
as defined in \cite{Berger1999}, we associate for any $f\in GL_2(\mathbb{F}_q)$ a function $u_f: \mathbb{P}^1(\mathbb{F}_q) \rightarrow \mathbb{F}_q$
\begin{equation}
	u_f(z)=
	\begin{cases}
		cz+d, & \text{if}\ z\in\mathbb{F}_q  ~and~ cz+d\neq 0 \\
		az+b, & \text{if}\ z\in\mathbb{F}_q ~and~ cz+d= 0 \\
		c, & \text{if}\ z=\infty ~and~ c\neq 0 \\ 
		a, & \text{if}\ z=\infty ~and~ c= 0 \\
	\end{cases}
\end{equation} 

For a matrix $f=(f_{il})\in GL_2(\mathbb{F}_q)$ and an integer $e\geq 0$, we denote by $f^{(p^e)}=\left(f_{il}^{\,p^{e}}\right)$ the matrix obtained by applying the Frobenius $x\mapsto x^{p^e}$ to the entries of $f$. In particular, $f^{(p^{m-j})}$ is obtained by applying the inverse of $\gamma_{p^j}:x\mapsto x^{p^j}$ entrywise.

\begin{proposition} \label{theo_ufg}
Let $f,g\in GL_2(\mathbb F_q)$ and every
$0\le j\le m-1$, the following identities hold:
\begin{align}
u_{fg}(z)
&=
u_f(\overline g(z))\,u_g(z),
\label{eq:ufg}
\\
u_f(\gamma_{p^j}(z))
&=
\gamma_{p^j}\!\left(
u_{f^{(p^{m-j})}}(z)
\right).
\label{eq:ufrobenius}
\end{align}    
\end{proposition}
\begin{proof}
For the identity \eqref{eq:ufg}, see \cite{DUR198769}. We now prove the relation \eqref{eq:ufrobenius}.\\
    Suppose $f=\left(
\begin{array}{cc}
	a & b \\
	c & d \\
\end{array}
\right) \in GL_2(\mathbb{F}_q)$,we have $\gamma_{p^j}(x)=x^{p^j}$.\\
\begin{align}
    u_f(\gamma_{p^j}(z)) &=
	\begin{cases}
		cz^{p^j} + d, & \text{if}\ z^{p^j} \in\mathbb{F}_q  ~and~ cz^{p^j} + d \neq 0 \\
		az^{p^j} + b, & \text{if}\ z^{p^j} \in\mathbb{F}_q ~and~ cz^{p^j} + d = 0 \\
		c, & \text{if}\ z^{p^j}=\infty ~and~ c\neq 0 \\ 
		a, & \text{if}\ z^{p^j}=\infty ~and~ c= 0 \\
	\end{cases} \\
    & = \begin{cases}
		cz^{p^j} + d, & \text{if}\ z^{p^j} \in\mathbb{F}_q  ~and~ (c^{p^{m-j}}z + d^{p^{m-j}})^{p^j} \neq 0 \\
		az^{p^j} + b, & \text{if}\ z^{p^j} \in\mathbb{F}_q ~and~ (c^{p^{m-j}}z + d^{p^{m-j}})^{p^j} = 0 \\
		(c^{p^{m-j}})^{p^j}, & \text{if}\ z^{p^j}=\infty ~and~ (c^{p^{m-j}})^{p^j} \neq 0 \\ 
		(a^{p^{m-j}})^{p^j}, & \text{if}\ z^{p^j}=\infty ~and~  (c^{p^{m-j}})^{p^j} = 0 \\
	\end{cases} \\
    &= 
    \begin{cases}
		\gamma_{p^j}(c^{p^{m-j}}z + d^{p^{m-j}}), & \text{if}\ z \in\mathbb{F}_q  ~and~ c^{p^{m-j}}z + d^{p^{m-j}} \neq 0 \\
		\gamma_{p^j}(a^{p^{m-j}}z + b^{p^{m-j}}), & \text{if}\ z \in\mathbb{F}_q ~and~ c^{p^{m-j}}z + d^{p^{m-j}} = 0 \\
		\gamma_{p^j}(c^{p^{m-j}}), & \text{if}\ z=\infty ~and~ c^{p^{m-j}} \neq 0 \\ 
		\gamma_{p^j}(a^{p^{m-j}}), & \text{if}\ z=\infty ~and~  c^{p^{m-j}} = 0 \\
	\end{cases}
     \\
    &= \gamma_{p^j}(u_{f^{(p^{m-j})}}(z))
\end{align}
\end{proof}

\begin{theorem}[\cite{DUR198769}, Theorem 4] \label{dur} The Cauchy codes $C_k(x,y)$ and $C_k(\alpha,\beta)$ are equal
	if there exist $f \in GL_2(\mathbb{F}_q)$, $ \lambda \in \mathbb{F}_q^{*}$ such that $\alpha_i=\overline{f}(x_i)$ and $\beta_i=\lambda u_f(x_i)^{k-1}y_i$.\\
\end{theorem}

While Theorem~\ref{dur} characterizes equivalence between Cauchy codes via projective linear transformations, the description of their automorphism groups requires the more general framework of projective semilinear transformations. Following Berger \cite{Berger1999}, we therefore introduce the notion of a permutation induced by an element of $P\Gamma L_2(\mathbb{F}_q)$.
\begin{definition}
Let $C_k(x,y)$ be a Cauchy code with support $L_x=\{x_0,\ldots,x_{n-1}\}$,
where $n\ge3$, and let $F=\gamma_{p^j}\circ\overline f
\in
P\Gamma L_2(\mathbb F_q)$. We say that $F$ induces a permutation on $L_x$ if $F(L_x)=L_x.$
The induced permutation $\sigma\in S_n$ is then defined by $F(x_i)=x_{\sigma(i)},\: 0\le i<n.$
\end{definition}

\begin{theorem}[\cite{Berger1999}, Theorem 1.11] \label{berger} Let $C_k(x,y)$ be a Cauchy code and $2\leq k\leq n-1$. 
	$(e,\sigma, \gamma_{p^j}) \in \Gamma Aut(C_k(x,y))$ if and only if there exists  $F=\gamma_{p^j} \circ \overline{f} \in P\Gamma L_2(\mathbb{F}_q)$ such that:
	\begin{itemize}
		\item $F(L_x)=L_x$ and $\sigma$ induced by $F$.
		\item there exists $\lambda \in \mathbb{F}_q^{*}$, such that 
		$e_i=\lambda y_{\sigma^{-1}(i)}^{p^{m-j}}(u_f(x_i)^{k-1}y_i)^{-1}$
	\end{itemize}

\end{theorem}

A direct computation on $\frac{az+b}{cz+d}$ shows that, for every $h\in GL_2(\mathbb{F}_q)$,
\begin{equation}\label{eq:twist}
\gamma_{p^j}\circ\overline{h}\circ\gamma_{p^j}^{-1}=\overline{h^{(p^{j})}}
\qquad\text{and}\qquad
\gamma_{p^j}^{-1}\circ\overline{h}\circ\gamma_{p^j}=\overline{h^{(p^{m-j})}} .
\end{equation}
We introduce the notion of $\gamma_{p^j}$-similarity to describe the equivalence induced by projective changes of coordinates on semilinear automorphisms.
\begin{definition}\label{def:similar}
Let $0\le j\le m-1$. Define an action of $GL_2(\mathbb{F}_q)$ on itself by $h\cdot_j f=h^{(p^{m-j})}fh^{-1}.$
The corresponding orbits are called the $\gamma_{p^j}$-similarity classes. Equivalently, two matrices $f,g\in GL_2(\mathbb{F}_q)$ are said to be $\gamma_{p^j}$-similar if there exists $h\in GL_2(\mathbb{F}_q)$ such that $g=h^{(p^{m-j})}fh^{-1}.$
\end{definition}

The $\gamma_{p^j}$-similarity relation generalizes the classical matrix similarity, which is recovered when $j=0$. If $h$ is replaced by $\mu h$ with $\mu\in\mathbb{F}_q^{*}$, then $h\cdot_j f$ is multiplied by the scalar $\mu^{p^{m-j}-1}$. So the same action is defined on $PGL_2(\mathbb{F}_q)$, and in the sequel the classes are considered in $PGL_2(\mathbb{F}_q)$.

\begin{proposition}
The $\gamma_{p^j}$-similarity is an equivalence relation on $GL_2(\mathbb{F}_q)$.
\end{proposition}

\begin{proof}
Since $(h_1h_2)^{(p^{m-j})}=h_1^{(p^{m-j})}h_2^{(p^{m-j})}$, we have $h_1\cdot_j(h_2\cdot_j f)=(h_1h_2)\cdot_j f$, so $\cdot_j$ is an action. The $\gamma_{p^j}$-similarity is the orbit relation associated with this action of $GL_2(\mathbb{F}_q)$ on itself.
\end{proof}
\begin{theorem}[Twisted Conjugacy Correspondence]\label{lem:semilinear-conjugacy}
Let $0\le j\le m-1$. For every homography
$f\in PGL_{2}(\mathbb{F}_q)$, define the semilinear transformation $P(f)=\gamma_{p^j} \circ f\in P\Gamma L_{2}(\mathbb{F}_q).$

Then, for all $f,g,h\in PGL_{2}(\mathbb{F}_q)$, we have
\[
g=h^{(p^{m-j})}\circ f\circ h^{-1}
\qquad\Longleftrightarrow\qquad
P(g)=h\circ P(f)\circ h^{-1}.
\]
\end{theorem}

\begin{proof}
By \eqref{eq:twist} we have $\gamma_{p^j}^{-1}\circ h\circ \gamma_{p^j}=h^{(p^{m-j})}$. We obtain
\[
\begin{aligned}
P(g)=h\circ P(f)\circ h^{-1}
&\iff
\gamma_{p^j} \circ g
=
h\circ \gamma_{p^j} \circ f\circ h^{-1} \\
&\iff
g
=
\gamma_{p^j}^{-1} \circ h\circ \gamma_{p^j} \circ f\circ h^{-1} \\
&\iff g = h^{(p^{m-j})} \circ f \circ h^{-1}.
\end{aligned}
\]
Consequently, the map $P:f\longmapsto \gamma_{p^j}\circ f$
is an equivariant bijection between
$PGL_2(\mathbb{F}_q)$ with the $\gamma_{p^j}$-similarity and the conjugation action of
$PGL_2(\mathbb{F}_q)$ on the coset
$\gamma_{p^j}.PGL_2(\mathbb{F}_q)$.
Hence it induces a bijection between the corresponding orbit spaces.
\end{proof}
\begin{remark}\label{rem:norm}
For $j=0$ the statement reduces to ordinary conjugacy in $PGL_2(\mathbb{F}_q)$.
For $j=m/2$ we have $\gamma_{p^j}^2=id$ and $h^{(p^{m-j})}=h^{(p^j)}$, and one moreover has
\[
P(f)^2=\gamma_{p^j}\circ f\circ\gamma_{p^j}\circ f=f^{(p^j)}\circ f = N(f),
\]
where $N(f)$ is the Shintani norm of $f$. We will see in Corollary~\ref{cor:involutions} that if $P(f)^2=id$ then $P(f)$ is conjugate to $\gamma_{p^{m/2}}$.
\end{remark}
\begin{remark}
Theorem~\ref{lem:semilinear-conjugacy} shows that the classification of
$\gamma_{p^j}$-similarity classes in $PGL_2(\mathbb{F}_q)$ can be translated into the
classification of ordinary conjugacy classes of the corresponding
elements of $P\Gamma L_2(\mathbb{F}_q)$. This allows us to apply the standard
conjugacy classification in the semilinear group to the twisted
conjugacy problem arising from induced centrosymmetric alternant
codes.
\end{remark}

\section{Induced Centrosymmetric Alternant Codes} \label{sec:centro}
Alternant codes are obtained as subfield subcodes of Cauchy (or generalized Reed--Solomon) codes. Therefore, the characterization of automorphism groups for Cauchy codes given by Theorem~\ref{berger} naturally extends to alternant codes through the subfield subcode construction. We first recall the definition of a subfield subcode.
\begin{definition}(Subfield Subcode)\\
	Let $\C \subset \F_{q}^n$ be an $[n,k,d]$-code over $\F_{q}$. The subfield subcode  of $\C$ on the subfield $\F_p$, denoted by $\C_{|\F_p}$, is the subcode of $\C$ whose elements are in $\F_p^n$.
	$$\C_{|\F_p}= \C \cap \F_p^n.$$
\end{definition}
An alternant code is defined as the subfield subcode of the dual of a GRS code. In \cite{Berger1999} the alternant code is the subfield subcode of the Cauchy code itself. Since the dual of a GRS code is a GRS code with the same support, the two conventions describe the same family of codes.
\begin{definition}(Alternant code)\\
	The Alternant code $A_k(x,y)$ associated to $(x,y)$ is defined by:
	$$ A_k(\textbf{x,y})=(GRS_k(x,y)^{\bot})_{|\F_p}=GRS_k(x,y)^{\bot}\cap \F_p^n$$
\end{definition} 
\begin{lemma}[\cite{Berger1999}, Proposition 1.6]\label{lem:sub} Suppose that there exists
	$(e,\sigma, \gamma) \in \Gamma Aut(\C)$ such that $e=\lambda(1,..,1)$ then $\sigma(\C_{|\F_p})=\C_{|\F_p}.$
\end{lemma}
\begin{proof}
Let $c\in \C_{|\F_p}$. Since the coordinates of $c$ are in $\F_p$, we have $(e,\sigma,\gamma)(c)=\gamma(\lambda)\,\sigma(c)$. This word is in $\C$ and $\C$ is linear, so $\sigma(c)\in\C$. As $\sigma(c)\in\F_p^n$, we get $\sigma(\C_{|\F_p})\subseteq\C_{|\F_p}$, and the equality holds because both sets have the same number of elements.
\end{proof}
In the following, $\sigma(i) = n-1-i$ is the centrosymmetric permutation and $J_k$ denotes the $k \times k$ anti-identity matrix, so that the
permutation matrix of $\sigma$ is $P_\sigma = J_n$.
\begin{definition}
A matrix $M=(m_{ij})\in\mathbb{F}_q^{k\times n}$ is called
\emph{centrosymmetric} if $J_k M J_n=M$. Equivalently, $m_{ij}=m_{k-1-i,\,n-1-j}$, for all $0\le i\le k-1$ and $0\le j\le n-1$.
\end{definition}
\begin{definition}
A code $C\subseteq\mathbb{F}_q^n$ is called \emph{centrosymmetric} if it is invariant under the permutation $\sigma$.
Equivalently, $\sigma(C)=C$.
\end{definition}
\begin{remark}
A centrosymmetric code does not necessarily have a centrosymmetric generator matrix. For example, codes of dimension $\geq 2$ that are generated by palindromic words.
\end{remark}
\begin{lemma}\label{lem:M}
Let $G$ be a generator matrix of $C$. Then $\sigma(C) = C$ if and only if there
exists $M \in \F_q^{k \times k}$ such that
\begin{equation}\label{eq:GJ=MG}
  G\, J_n = M\, G .
\end{equation}
In this case $M$ is unique and $M^2 = I_k$; it is the matrix of the restriction
of $\sigma$ to $C$ in the basis given by the rows of $G$.
\end{lemma}

\begin{proof}
The rows of $G J_n$ span $\sigma(C)$, so $\sigma(C) \subseteq C$ if and only if
each row of $G J_n$ is a linear combination of the rows of $G$, which is
\eqref{eq:GJ=MG}; equality follows since $\sigma$ preserves the dimension.
Uniqueness comes from $\operatorname{rank} G = k$. Finally
$M^2 G = M G J_n = G J_n^2 = G$, hence $M^2 = I_k$.
\end{proof}

\begin{proposition}\label{prop:centro_generator}
Assume $\sigma(C) = C$, let $G$ be a generator matrix of $C$ and $M$ as in
Lemma~\ref{lem:M}. The following assertions are equivalent.
\begin{enumerate}
  \item $C$ admits a centrosymmetric generator matrix.
  \item There exists an invertible $B \in \F_q^{k \times k}$ such that
        \begin{equation}\label{eq:BM=JB}
          B\, M = J_k\, B .
        \end{equation}
\end{enumerate}
Moreover, the centrosymmetric generator matrices of $C$ are exactly the matrices
$G_c = B\,G$ with $B$ invertible satisfying \eqref{eq:BM=JB}.
\end{proposition}

\begin{proof}
Every generator matrix of $C$ is of the form $B G$ with $B$ invertible, and by
\eqref{eq:GJ=MG}
\[
  J_k\,(B G)\,J_n = J_k B M G .
\]
Since $\operatorname{rank} G = k$, we get $J_k (BG) J_n = BG$ if and only if
$J_k B M = B$, i.e.\ $B M = J_k B$ (as $J_k^2 = I_k$).
\end{proof}


As in \cite{Berger1999} for the cyclic case, an alternant code $A_k(x,y)$ is called \emph{induced-centrosymmetric} by $\sigma$, if the
automorphism group of its underlying GRS code $GRS_k(x,y)$ contains an element
of the form $(\lambda\mathbf{1},\sigma,\gamma)$. Note that when $\gamma\neq id$ the code $GRS_k(x,y)$ itself is in general not invariant under $\sigma$. Only the alternant code is.

The next theorem is the analogue, for the centrosymmetric permutation, of \cite[Theorem 2.2]{Berger1999}, which treats the cyclic shift. Here the support is a union of orbits of length at most two instead of a single orbit.
\begin{theorem}\label{tho_3}
Let $A_k(x,y)$ be an alternant code of length $n$.

Then $A_k(x,y)$ is induced-centrosymmetric if and only if there exists an element
$F=\gamma_{p^j}\circ\overline{f}\in P\Gamma L_2(\mathbb{F}_q)$ satisfying the following conditions.

\begin{enumerate}
\item The support $x$ is the union of the $F$-orbits. More precisely,

\begin{itemize}
\item if $n$ is even, then $L_x=\bigcup_{i=0}^{\frac n2-1}\mathcal O_F(x_i);$

\item if $n$ is odd, then $L_x=\left(\bigcup_{i=0}^{\frac{n-3}2}\mathcal O_F(x_i)\right)
\cup
\{x_{\frac{n-1}2}\},$
where $F(x_{\frac{n-1}2})=x_{\frac{n-1}2}$.
\end{itemize}

Equivalently, $x_{n-1-i}=F(x_i),
\qquad
0\le i<n.$

\item The multiplier satisfies $y_{n-1-i}
=
\left(
\lambda\,u_f(x_i)^{\,k-1}y_i
\right)^{p^j},
\qquad
0\le i<n,$
for some $\lambda\in\mathbb{F}_q^\times$.
\end{enumerate}
\end{theorem}

  \begin{proof}
Let $\mathcal{A}_k(\boldsymbol{x},\boldsymbol{y})$ be an induced-centrosymmetric alternant code of length $n$.
By the definition of induced-centrosymmetric, there exists $(\lambda' \mathbf{1},\sigma,\gamma_{p^j}) \in \Gamma\operatorname{Aut} (\operatorname{GRS}_k(\boldsymbol{x},\boldsymbol{y}))$, where $\lambda'\in\mathbb{F}_q^*$.

Applying Theorem~\ref{berger}, there exists $F=\gamma_{p^j}\circ\overline{f}\in P\Gamma L_2(\mathbb{F}_q)$
such that $F$ induces the permutation $\sigma$ on the support
$L_{\boldsymbol{x}}$. Thus $F(x_i)=x_{\sigma(i)}=x_{n-1-i}, \qquad 0\leq i\leq n-1$,
and $F(L_{\boldsymbol{x}})=L_{\boldsymbol{x}}$.

Moreover, by Theorem~\ref{berger}, we have $\lambda'= \lambda'' y_{\sigma^{-1}(i)}^{p^{m-j}}\left(u_f(x_i)^{k-1}y_i\right)^{-1}$ where $\lambda''\in\mathbb{F}_q^*$.
Since $\sigma^{-1}=\sigma$, this becomes $\lambda=y_{n-1-i}^{p^{m-j}} \left(u_f(x_i)^{k-1}y_i\right)^{-1}$ where $\lambda = \frac{\lambda'}{\lambda''}$,
or equivalently $y_{n-1-i}^{{p^{m-j}}}=\lambda u_f(x_i)^{k-1}y_i \implies y_{n-1-i}=(\lambda u_f(x_i)^{k-1}y_i)^{{p^{j}}}$.

Conversely, suppose there exists a projective semilinear transformation \(F=\gamma_{p^j}\circ\overline{f}\in P\Gamma L_2(\mathbb F_q)\) such that $x_{n-1-i}=F(x_i)
\text{ and } y_{n-1-i}=\left(\lambda\,u_f(x_i)^{k-1}y_i\right)^{p^j}$,
for each \(0\le i<n\), where \(\lambda\in\mathbb F_q^\times\).

Let $C=\operatorname{GRS}_{k}(x,y)$, so that $ A_k(x,y)=C^{\bot}\cap\mathbb F_p^n$.
Consider a codeword $c_i=y_i P(x_i) \in C$,
where \(P\in\mathbb F_q[X]\) satisfies \(\deg P<k\). Then, for each \(0\le i<n\),
\[
\begin{aligned}
c_{n-1-i}
&=y_{n-1-i}P(x_{n-1-i})\\
&=\left(\lambda\,u_f(x_i)^{k-1}y_i\right)^{p^j}
P\bigl(\gamma_{p^j}(\overline f(x_i))\bigr).
\end{aligned}
\]
Indeed, by writing $\overline f(X)=\frac{aX+b}{cX+d},
\: ad-bc\ne0$, we have \[
c_{n-1-i} = \left(\lambda\,u_f(x_i)^{k-1}y_i P^{(p^{-j})}( \frac{ax_i + b}{cx_i + d})  \right)^{p^j}
\] 
where \[
\big(\sum_{i=0}^{d} a_i X^i \big)^{(p^{-j})} = \sum_{i=0}^{d} a_i^{p^{m-j}} X^i
\]
Then $c_{n-1-i}=\bigl(y_iQ(x_i)\bigr)^{p^j}$,
where \(Q(x)= \lambda u_f(x)^{k-1}P^{(p^{-j})}( \frac{ax + b}{cx + d})  \in\mathbb F_q[X]\) has a degree less than \(k\), therefore $\sigma(c)=w^{p^j} = (w_0^{p^j},\cdots,w_{n-1}^{p^j})$ with $w \in C$. 

Let $\gamma_{p^j}(C) = \{ c^{p^j}, c\in C \}$ and then $\sigma(C) = \gamma_{p^j}(C)$. In particular $(\mathbf 1,\sigma,\gamma_{p^j}^{-1})\in\Gamma Aut(C)$, so the code is induced-centrosymmetric. Since $\langle\sigma(a),\sigma(b)\rangle=\langle a,b\rangle$ and $\langle\gamma_{p^j}(a),\gamma_{p^j}(b)\rangle=\gamma_{p^j}(\langle a,b\rangle)$, the maps $\sigma$ and $\gamma_{p^j}$ both preserve orthogonality, and we also have $\sigma(C^{\bot}) = \gamma_{p^j}(C^{\bot})$.

Therefore \[
\begin{aligned}
    \sigma(A_k(x,y)) &= \sigma(C^{\bot} \cap \F_p^n) \\
    &= \sigma(C^{\bot}) \cap \F_p^n \\
    &= \gamma_{p^j}(C^{\bot}) \cap \F_p^n \\
    &= \gamma_{p^j}(C^{\bot} \cap \F_p^n) \\
    &= C^{\bot} \cap \F_p^n \\
    &= A_k(x,y)
\end{aligned}
\]
Here we used that $\gamma_{p^j}$ fixes every element of $\F_p^n$. This proves the reverse implication.
\end{proof}

Proposition 2.3 of \cite{Berger1999} states the next result for the support, without proof. We give a complete proof and the corresponding relation for the multiplier.
\begin{theorem}\label{lemmegamma}
Let $\mathcal{C}=GRS_k(x;y)$
be a generalized Reed--Solomon code and let
$\sigma\in S_n$. Assume that there exist $f\in GL_2(\mathbb{F}_q)$, $0\le j\le m-1$,
and $\lambda\in\mathbb{F}_q^\times$  such that $x_{\sigma^{-1}(i)} = (\gamma_{p^j}\circ\overline{f})(x_i),$
and $y_{\sigma^{-1}(i)} = \left(\lambda\,u_f(x_i)^{k-1}y_i\right)^{p^j}$,
for every $i$. Let $g \in GL_2(\mathbb{F}_q)$ such that $g = h^{(p^{m-j})} \circ f \circ h^{-1}$ with $h\in GL_2(\mathbb{F}_q)$.\\
Then there exist $u$ and $v$, support and multiplier  such that
$$u_{\sigma^{-1}(i)}=(\gamma_{p^j}\circ\overline{g})(u_i) ,\: v_{\sigma^{-1}(i)}=\left(\lambda\,u_g(u_i)^{k-1}v_i\right)^{p^j}
~and~  GRS_k(x;y)=GRS_k(u;v) $$
\end{theorem}

\begin{proof}
Suppose $h=\begin{pmatrix}
\alpha&\beta\\
\gamma&\delta
\end{pmatrix}
\in GL_2(\mathbb{F}_q)$ and $u_i=\overline h(x_i)$, $v_i = u_h(x_i)^{k-1}y_i $ for $\qquad i=1,\ldots,n$.

By the invariance of generalized Reed--Solomon codes under projective linear transformations, there exists a vector of multipliers
$v=(v_1,\ldots,v_n)$ such that
\[
\mathcal C=GRS_k(u;v).
\]
Here $u$ may contain the point $\infty$ (when $\overline h(x_i)=\infty$). In that case $GRS_k(u;v)$ is understood as the Cauchy code $C_k(u,v)$.
Then we have \\
\begin{equation}
\begin{split}
u_{\sigma^{-1}(i)}
&=\overline h(x_{\sigma^{-1}(i)})\\
&=\overline h\!\left((\gamma_{p^j}\circ\overline f)(x_i)\right).\\
&= \frac{
\alpha(ax_i+b)^{p^j}
+\beta(cx_i+d)^{p^j}
}{
\gamma(ax_i+b)^{p^j}
+\delta(cx_i+d)^{p^j}
}\\
&= \left(
\frac{
\alpha^{(p^{m-j})}(ax_i+b)
+\beta^{(p^{m-j})}(cx_i+d)
}{
\gamma^{(p^{m-j})}(ax_i+b)
+\delta^{(p^{m-j})}(cx_i+d)
}
\right)^{p^j}.\\
&= \gamma_{p^j} \circ h^{(p^{m-j})} \circ \overline{f}(x_i)
\end{split}
\end{equation}
The relation $g =h^{(p^{m-j})} \circ f \circ h^{-1} \implies  h^{(p^{m-j})} \circ \overline{f} = g \circ \overline{h}$
and therefore
\[
u_{\sigma^{-1}(i)}
=
(\gamma_{p^j}\circ\overline g \circ \overline{h})(x_i) = \gamma_{p^j}\circ\overline g(u_i)
\]
\begin{equation}
\begin{split}
v_{\sigma^{-1}(i)}
&= u_h(x_{\sigma^{-1}(i)})^{k-1}y_{\sigma^{-1}(i)} \\
&= u_{h} (\gamma_{p^j} \circ \overline{f}(x_i))^{k-1} (\lambda u_f(x_i)^{k-1} y_i)^{p^j} \\
&= \gamma_{p^j} \circ u_{h^{(p^{m-j})}} (\overline{f}(x_i))^{k-1} (\lambda u_f(x_i)^{k-1} y_i)^{p^j} \text{ by using the proposition } \ref{theo_ufg} \\
&= \gamma_{p^j}(u_{h^{(p^{m-j})}}(\overline{f}(x_i) u_f(x_i))^{k-1} \lambda y_i) \text{ with the relation } \ref{eq:ufg} \text{ holds } \\
&= \gamma_{p^j} ( u_{h^{(p^{m-j})} \circ f} (x_i)^{k-1} \lambda y_i ) \text{ since } g \circ h = h^{(p^{m-j})} \circ f \text{ we have } \\
&=  \gamma_{p^j}( u_{g \circ h}(x_i)^{k-1} \lambda y_i ) \text{ with also the proposition } \ref{eq:ufg} \\
&= \gamma_{p^j}( u_g(\overline{h}(x_i))^{k-1} u_h(x_i)^{k-1} \lambda y_i ) \\
&= \gamma_{p^j}( \lambda u_g(u_i)^{k-1} v_i) \\
&= (\lambda u_g(u_i)^{k-1} v_i)^{p^j}
 \end{split}
\end{equation}
which completes the proof.
\end{proof}
A natural question is how to construct centrosymmetric alternant codes,
starting from the various possible forms of the matrix $f$. Theorem~\ref{lemmegamma} shows that
if two matrices $f$ and $g$ in $GL_2(\F_q)$ are $\gamma_{p^j}$-similar, then a GRS code whose automorphism is induced by $f$ admits another support and another multiplier for which the automorphism is induced by $g$. So it is enough to study $f$ up to $\gamma_{p^j}$-similarity.

In the rest of this section, $\sigma(i)=n-1-i$ and $n\geq 3$. We consider an induced centrosymmetric alternant code $A_k(x,y)$. By Theorem~\ref{tho_3} there exist $F=\gamma_{p^j}\circ\overline f$ and $\lambda\in\F_q^{*}$ such that
\begin{equation}\label{eq:rel}
x_{n-1-i}=F(x_i),\qquad y_{n-1-i}=\left(\lambda\,u_f(x_i)^{k-1}y_i\right)^{p^j},\qquad 0\le i<n .
\end{equation}

\subsection{Reduction to the cases $j=0$ and $j=\frac m2$}
The second point of the next lemma is the first case of \cite[Theorem 2.4]{Berger1999}. We recall the proof.

\begin{lemma}\label{lem:rigid}
\begin{enumerate}
\item A homography of $\mathbb{P}^1(\F_q)$ which has at least three fixed points is the identity.
\item Let $\varphi$ be an automorphism of $\F_q$ with fixed field $L$, and let $Q=\varphi\circ\overline g$ with $g\in GL_2(\F_q)$. If $Q$ has at least three fixed points in $\mathbb{P}^1(\F_q)$, then there exists $h\in GL_2(\F_q)$ such that $Q=\overline h\circ\varphi\circ\overline h^{-1}$. In particular the set of fixed points of $Q$ is $\overline h(\mathbb{P}^1(L))$.
\end{enumerate}
\end{lemma}
\begin{proof}
1. The fixed points of $z\mapsto\frac{az+b}{cz+d}$ are the roots of $cz^2+(d-a)z-b$, and $\infty$ if $c=0$. If there are three fixed points then $c=0$, and the polynomial $(d-a)z-b$ has at least two roots. Hence $a=d$ and $b=0$.

2. Let $h$ be a matrix such that $\overline h$ sends $0,1,\infty$ to three fixed points of $Q$. As in \eqref{eq:twist}, the map $Q'=\overline h^{-1}\circ Q\circ\overline h$ can be written $Q'=\varphi\circ\overline{g'}$ with $g'\in GL_2(\F_q)$. The points $0,1,\infty$ are fixed by $Q'$ and by $\varphi$, so they are fixed by $\overline{g'}$. By the first point $\overline{g'}$ is the identity and $Q'=\varphi$. The fixed points of $\varphi$ are the points of $\mathbb{P}^1(L)$.
\end{proof}

The next proposition shows that there is no loss of generality in taking $j\in\{0,\frac m2\}$.

\begin{proposition}\label{prop:reduction}
Assume \eqref{eq:rel} with $0\le j\le m-1$ and let $L=\F_{p^{d}}$ with $d=\gcd(2j,m)$ be the fixed field of $\gamma_{p^j}^2$. Then there exist a support $u$ and a multiplier $v$ such that $GRS_k(x;y)=GRS_k(u;v)$, with $u_i\in\mathbb{P}^1(L)$ and $v_i\in L^{*}$ for all $i$, and \eqref{eq:rel} holds for $(u,v)$ with a matrix $g\in GL_2(L)$ instead of $f$. In particular the code is defined over $L$, and the restriction of $\gamma_{p^j}$ to $L$ has order $1$ or $2$.
\end{proposition}
\begin{proof}
Put $s=p^j$. Since $\sigma^2=id$, the map $Q=F\circ F$ fixes every $x_i$. By \eqref{eq:twist} we can write $Q=\gamma_{p^j}^2\circ\overline{g_0}$. As $n\geq3$, Lemma~\ref{lem:rigid} gives a matrix $h$ such that $Q=\overline h\circ\gamma_{p^j}^2\circ\overline h^{-1}$, and all the $x_i$ are in $\overline h(\mathbb{P}^1(L))$. Let $u_i=\overline h^{-1}(x_i)\in\mathbb{P}^1(L)$. By Theorem~\ref{lemmegamma} there is a multiplier $v$ with $GRS_k(x;y)=GRS_k(u;v)$, and \eqref{eq:rel} holds for $(u,v)$ with $F'=\overline h^{-1}\circ F\circ\overline h=\gamma_{p^j}\circ\overline g$.

We have $F'\circ F'=\gamma_{p^j}^2$. So $F'$ commutes with $\gamma_{p^j}^2$ and sends $\mathbb{P}^1(L)$ into itself. The same is true for $\gamma_{p^j}$, hence $\overline g$ sends $\mathbb{P}^1(L)$ into itself. A homography is determined by the images of $0,1,\infty$, so we can take $g$ with entries in $L$. Moreover $F'\circ F'=\gamma_{p^j}^2\circ\overline{g^{(p^{m-j})}g}$, hence $g^{(p^{m-j})}g=\mu I$ with $\mu\in L^{*}$.

For the multiplier, we apply the second relation of \eqref{eq:rel} twice. Since $z^{s^2}=z$ for $z\in L$, we get
\[
v_i=\lambda^{s+s^2}\left(u_g\big(\gamma_{p^j}(\overline g(u_i))\big)^{s}\,u_g(u_i)\right)^{k-1}v_i^{s^2}.
\]
By Proposition~\ref{theo_ufg}, $u_g(\gamma_{p^j}(w))^{s}=u_{g^{(p^{m-j})}}(w)$ for $w\in\mathbb{P}^1(L)$, and
\[
u_{g^{(p^{m-j})}}(\overline g(u_i))\,u_g(u_i)=u_{g^{(p^{m-j})}g}(u_i)=\mu .
\]
So $v_i^{s^2-1}=\lambda^{-(s+s^2)}\mu^{-(k-1)}$ does not depend on $i$. Then $(v_i/v_0)^{s^2-1}=1$, which means $v_i/v_0\in L^{*}$. Dividing $v$ by $v_0$ does not change the code.
\end{proof}

\begin{remark}
If $\gamma_{p^j}$ has order greater than $2$, Proposition~\ref{prop:reduction} shows that the support is, up to a homography, in the subfield $L$. In that case $n\le p^{\gcd(2j,m)}+1$, and we are in the situation $j\in\{0,\frac d2\}$ over $L$.
\end{remark}

In the constructions that follow, we restrict ourselves to the cases \(j\in\{0,m/2\}\), where \(m\) is assumed to be even when $j=m/2$.

\begin{lemma}\label{lem:Nid}
Let $j\in\{0,\frac m2\}$ and $n \ge 3$ . If $x_{n-1-i}=F(x_i)$ for $0\le i<n$, then $F\circ F=id$ on $\mathbb{P}^1(\F_q)$, that is $N(f)=f^{(p^j)}f$ is the identity of $PGL_2(\F_q)$.
\end{lemma}
\begin{proof}
By Remark~\ref{rem:norm}, $F\circ F=\overline{N(f)}$ is a homography. It fixes the $n\geq 3$ points $x_i$, so it is the identity by Lemma~\ref{lem:rigid}.
\end{proof}

\subsection{Classification of the $\gamma_{p^j}$-similarity classes}

For \(j=m/2\), the classification can be understood through the Shintani norm. Indeed, the map \(P(f)=\gamma_{p^j}\circ f\) transforms the \(\gamma_{p^j}\)-twisted conjugacy relation into ordinary conjugacy under \(\operatorname{PGL}_2(\F_q)\) within the semilinear coset \(\gamma_{p^j}\operatorname{PGL}_2(\F_q)\).
Moreover, since \(\gamma_{p^j}\) is an involution, we have $P(f)^2=f^{(p^j)}\circ f$,
which is the Shintani norm $N(f)$ of \(f\). If $g=h^{(p^j)}fh^{-1}$ then $N(g)=hN(f)h^{-1}$, so the conjugacy class of $N(f)$ and the number of fixed points of $P(f)$ only depend on the class of $f$.

On the other hand, for \(j=0\), the corresponding classification reduces to ordinary projective conjugacy in \(\operatorname{PGL}_2(\F_q)\). In both cases, the relevant projective transformations can be organized into four structural types: identity, split semisimple, parabolic (nontrivial unipotent), and elliptic (nonsplit semisimple).

When $j=m/2$ we put $r=p^{m/2}$, and we denote by $N_r(z)=z^{r+1}$ and $T_r(z)=z^r+z$ the norm and the trace of $\F_q$ over $\F_r$.

\begin{lemma}\label{lem:h90}
\begin{enumerate}
\item $N_r(u)=1$ if and only if $u=v^{r-1}$ for some $v\in\F_q^{*}$. The map $N_r:\F_q^{*}\to\F_r^{*}$ is surjective.
\item $T_r(y)=0$ if and only if $y=x^{r}-x$ for some $x\in\F_q$. The map $T_r:\F_q\to\F_r$ is surjective.
\end{enumerate}
\end{lemma}
\begin{proof}
The map $v\mapsto v^{r-1}$ is a morphism of $\F_q^{*}$ with kernel $\F_r^{*}$, so its image has $r+1$ elements. Since $N_r(v^{r-1})=v^{r^2-1}=1$, this image is contained in the kernel of $N_r$, which has at most $r+1$ elements. So the two sets are equal, and the image of $N_r$ has $r-1$ elements. In the same way $x\mapsto x^r-x$ is additive with kernel $\F_r$, its image has $r$ elements and is contained in the kernel of $T_r$, which has at most $r$ elements.
\end{proof}

The four types below were obtained in \cite[Theorem 2.4]{Berger1999} for every $j$, by counting the fixed points. 

\begin{theorem} \label{classification_grs}
\begin{enumerate}
\item Let $j=\frac m2$ and $r=p^{m/2}$. Every $f\in PGL_2(\F_q)$ is $\gamma_{p^j}$-similar to exactly one of the following matrices.
\begin{enumerate}
\item Identity: $f=\left(
		\begin{array}{cc}
			1 & 0 \\
			0 & 1 \\
		\end{array}
		\right)$. Then $N(f)=id$ and $P(f)$ has $r+1$ fixed points.
\item Parabolic: $f=\left(
		\begin{array}{cc}
			1 & b \\
			0 & 1 \\
		\end{array}
		\right)$ with $T_r(b)\neq 0$. All these matrices are in the same class. Then $\overline{N(f)}(z)=z+T_r(b)$ and $P(f)$ has one fixed point.
\item Split semisimple: $f=\left(
		\begin{array}{cc}
			a & 0 \\
			0 & 1 \\
		\end{array}
		\right)$ with $\nu=N_r(a)\neq 1$. The class only depends on the set $\{\nu,\nu^{-1}\}$. Then $\overline{N(f)}(z)=\nu z$ and $P(f)$ has two fixed points.
\item Elliptic: $f=\left(
		\begin{array}{cc}
			0 & \delta \\
			1 & 0 \\
		\end{array}
		\right)$ with $\zeta=\delta^{r-1}\neq 1$. The class only depends on the set $\{\zeta,\zeta^{-1}\}$. Then $\overline{N(f)}(z)=\zeta z$ and $P(f)$ has no fixed point.
\end{enumerate}
There are $r+2$ classes if $p$ is odd and $r+1$ classes if $p=2$.
\item Let $j=0$. Every $f\in PGL_2(\F_q)$ is conjugate to exactly one of the following: the identity; $z\mapsto z+1$; $z\mapsto az$ with $a\neq 0,1$, up to the exchange of $a$ and $a^{-1}$; $z\mapsto\frac{-d}{z+t}$ where $X^2-tX+d$ is irreducible over $\F_q$, up to the value of $t^2/d$.
\end{enumerate}
\end{theorem}
\begin{proof}
The second point is the classical classification of $PGL_2(\F_q)$ according to the number of fixed points on $\mathbb{P}^1(\F_q)$. We prove the first point.

Let $A\in GL_2(\F_q)$ be a matrix of $f$ and $M=A^{(r)}A$. We have $\det M=N_r(\det A)\in\F_r$ and $\operatorname{tr}(M)^{r}=\operatorname{tr}(AA^{(r)})=\operatorname{tr}(M)$. So the characteristic polynomial of $M$ has its coefficients in $\F_r$. It has degree $2$, hence its roots are in $\F_q$. Therefore $N(f)$ is the identity, or it has exactly one fixed point, or exactly two fixed points on $\mathbb{P}^1(\F_q)$. Note also that $P(f)$ commutes with $N(f)=P(f)^2$, so $P(f)$ permutes the fixed points of $N(f)$. In what follows, conjugation by a homography $h$ means that we replace $f$ by $h^{(r)}fh^{-1}$, that is $P(f)$ by $h\circ P(f)\circ h^{-1}$.

\emph{Case $N(f)=id$.} We have $A^{(r)}A=\mu I$. Applying the Frobenius gives $AA^{(r)}=\mu^{r}I$, and $\mu A=A(A^{(r)}A)=(AA^{(r)})A=\mu^{r}A$, so $\mu\in\F_r^{*}$. By Lemma~\ref{lem:h90} we can multiply $A$ by a scalar and assume $A^{(r)}A=AA^{(r)}=I$. Let $T(w)=Aw^{(r)}$ for $w\in\F_q^2$. The map $T$ is additive, $T(cw)=c^{r}T(w)$ and $T\circ T=id$. Take $\theta\in\F_q\setminus\F_r$. For every $w$, the vectors $w_1=w+T(w)$ and $w_2=\theta w+\theta^{r}T(w)$ are fixed by $T$ and $w=(w_2-\theta^{r}w_1)/(\theta-\theta^{r})$. So the vectors fixed by $T$ span $\F_q^2$. Let $B$ be a matrix whose columns are two independent fixed vectors. Then $AB^{(r)}=B$. With $h=B^{(r)}$ we obtain $A=h^{(r)}h^{-1}$, that is, $f$ is $\gamma_{p^j}$-similar to the identity. Moreover $P(f)=h\circ\gamma_{p^j}\circ h^{-1}$ has $r+1$ fixed points.

\emph{Case of one fixed point $z_0$.} Then $P(f)(z_0)=z_0$. After conjugation by a homography sending $z_0$ to $\infty$, we have $\overline f(z)=az+b$ and $\overline{N(f)}(z)=N_r(a)z+a^{r}b+b^{r}$. This map has no fixed point in $\F_q$, so $N_r(a)=1$. Writing $a^{-1}=w^{r-1}$ and conjugating by $z\mapsto wz$ we get $\overline f(z)=z+b$ with $T_r(b)\neq 0$. The conjugation by $z\mapsto z+c$ replaces $b$ by $b+c^{r}-c$, and the conjugation by $z\mapsto wz$ with $w\in\F_r^{*}$ multiplies $T_r(b)$ by $w$. By Lemma~\ref{lem:h90} we obtain only one class.

\emph{Case of two fixed points.} We send them to $0$ and $\infty$. If $P(f)$ fixes both points then $\overline f(z)=az$ and $\overline{N(f)}(z)=N_r(a)z$. If $N_r(a')=N_r(a)$ then $a'/a=v^{r-1}$ and the conjugation by $z\mapsto vz$ sends $az$ to $a'z$. The conjugation by $z\mapsto 1/z$ sends $az$ to $a^{-1}z$. If $P(f)$ exchanges the two points then $\overline f(z)=\delta/z$ and $\overline{N(f)}(z)=\delta^{r-1}z$. The conjugation by $z\mapsto wz$ replaces $\delta$ by $N_r(w)\delta$, so only $\zeta=\delta^{r-1}$ matters, and $z\mapsto1/z$ replaces $\zeta$ by $\zeta^{-1}$.

The four types are distinguished by the class of $N(f)$ and by the number of fixed points of $P(f)$. For the number of classes, there are $r-2$ possible values of $\nu\in\F_r^{*}\setminus\{1\}$ and $r$ possible values of $\zeta$ with $\zeta^{r+1}=1$, $\zeta\neq1$. These values are grouped by pairs $\{\nu,\nu^{-1}\}$ and $\{\zeta,\zeta^{-1}\}$. If $p$ is odd, the only value equal to its inverse is $-1$ in both cases, which gives $\frac{r-1}{2}$ split semisimple classes and $\frac{r+1}{2}$ elliptic classes, hence $1+1+\frac{r-1}{2}+\frac{r+1}{2}=r+2$ classes. If $p=2$, no value is equal to its inverse, which gives $\frac{r-2}{2}$ and $\frac{r}{2}$ classes, hence $r+1$ classes.
\end{proof}

\begin{theorem} \label{aut_grs3}
Let $j \in \{ 0,\frac m2 \}$ and let $\mathcal{C}=\operatorname{GRS}_k(x;y)$ be a generalized
Reed--Solomon code admitting a nontrivial permutation
automorphism $\sigma\in S_n$. Then there exist
$f\in\operatorname{PGL}_2(\mathbb{F}_q)$ and
$\lambda\in\mathbb{F}_q^*$ such that $f$ has one of the
following four forms:
\begin{enumerate}
    \item $f=\begin{pmatrix}1&0\\0&1\end{pmatrix}$,
    and, for $0\leq i\leq n-1$,
    \[
    x_{\sigma^{-1}(i)}=x_i^{p^j},
    \qquad
    y_{\sigma^{-1}(i)}=(\lambda y_i)^{p^j}.
    \]

    \item $f=\begin{pmatrix}a&0\\0&1\end{pmatrix},
    \: N_{p^{\frac m2}}(a) \ne 1$,
    and, for $0\leq i\leq n-1$,
    \[
    x_{\sigma^{-1}(i)}=(ax_i)^{p^j},
    \qquad
    y_{\sigma^{-1}(i)}=(\lambda y_i)^{p^j}.
    \]

    \item $f=\begin{pmatrix}1&b\\0&1\end{pmatrix},
    \: T_{p^{\frac m2}}(b) \ne 0$,
    and, for $0\leq i\leq n-1$,
    \[
    x_{\sigma^{-1}(i)}=(x_i+b)^{p^j},
    \qquad
    y_{\sigma^{-1}(i)}
    =(\lambda a^{k-1}y_i)^{p^j},
    \]
    \item $f=\begin{pmatrix}0&1\\c&d\end{pmatrix},
    \: c\in\mathbb{F}_q^*$,
    and, for $0\leq i\leq n-1$,
    \[
    x_{\sigma^{-1}(i)}
    =\left(\frac{1}{cx_i+d}\right)^{p^j},\qquad y_{\sigma^{-1}(i)}
    =\bigl(\lambda u_f(x_i)^{k-1}y_i\bigr)^{p^j}
    \]
\end{enumerate}
These four forms may belong to the same conjugacy class,
depending on the norm of $a$, the trace of $b$, and
the value of $c^{k-1}$.
\end{theorem}

\begin{proof}
According to Theorem~\ref{lemmegamma}, there exist
$f\in\operatorname{PGL}_2(\mathbb{F}_q)$, and $\lambda\in\mathbb{F}_q^*$ where $j\in\{0,\frac m2\}$
such that
\[
x_{\sigma^{-1}(i)}
=\gamma_{p^j}\circ\overline f(x_i),
\qquad
y_{\sigma^{-1}(i)}
=\bigl(\lambda u_f(x_i)^{k-1}y_i\bigr)^{p^j}.
\]

To determine the possible forms of $f$, we consider the
associated semilinear transformation $P(f)=\gamma_{p^j}\circ f$.



Following the theorem~\ref{classification_grs}, we obtain
\begin{enumerate}
    \item $f=\begin{pmatrix}1&0\\0&1\end{pmatrix}$.
    Then $\overline f(z)=z,\: u_f(z)=1$.
    Consequently,
    \[
    x_{\sigma^{-1}(i)}=x_i^{p^j},
    \qquad
    y_{\sigma^{-1}(i)}=(\lambda y_i)^{p^j}.
    \]

    \item $f=\begin{pmatrix}a&0\\0&1\end{pmatrix}, \: N_{p^{\frac m2}}(a) \ne 1$.
    We have $\overline f(z)=az,\qquad u_f(z)=1$.
    Thus,
    \[
    x_{\sigma^{-1}(i)}=(ax_i)^{p^j},
    \qquad
    y_{\sigma^{-1}(i)}=(\lambda y_i)^{p^j}.
    \]

    \item $f=\begin{pmatrix}1&b\\0&1\end{pmatrix},
    \:  T_{p^{\frac m2}}(b) \ne 0$. Then
    \[
    \overline f(z)=z + b,
    \qquad u_f(z)=1.
    \] we obtain
    \[
    x_{\sigma^{-1}(i)}=(x_i+b)^{p^j},
    \qquad
    y_{\sigma^{-1}(i)}
    =(\lambda y_i)^{p^j}.
    \]

    \item $f=\begin{pmatrix}0&1\\c&d\end{pmatrix}, \: c = \frac 1\delta\in\mathbb{F}_q^*$.
    Then
    \[
    \overline f(z)=\frac{1}{cz+d},
    \qquad u_f(z)=cz+d.
    \]
    Hence,
    \[
    x_{\sigma^{-1}(i)}
    =\left(\frac{1}{cx_i+d}\right)^{p^j},
    \]
    and
    \[
    y_{\sigma^{-1}(i)}
    =\bigl(\lambda(cx_i+d)^{k-1}y_i\bigr)^{p^j}.
    \]
\end{enumerate}
\end{proof}
\begin{remark}
For $j=\frac m2$ the number of classes is the number of conjugacy classes of $PGL_2(\F_{p^{m/2}})$ ( See~\cite{shintani1976two}).
\end{remark}

\begin{corollary}\label{cor:involutions}
\begin{enumerate}
\item If $j=\frac m2$, then $P(f)$ is an involution if and only if $N(f)=id$, if and only if $f$ is $\gamma_{p^j}$-similar to the identity. In this case $P(f)$ is conjugate to $\gamma_{p^{m/2}}$ and it has $p^{m/2}+1$ fixed points.
\item If $j=0$ and $f\neq id$, then $f^2=id$ if and only if $f$ is conjugate to $z\mapsto -z$ or to $z\mapsto \varepsilon/z$ with $\varepsilon$ a non-square, when $p$ is odd, and to $z\mapsto z+1$ when $p=2$.
\end{enumerate}
\end{corollary}

\subsection{Explicit forms}
We now list the conditions on the support and on the multiplier for the usual forms of $F$.
\begin{theorem} \label{alternant_classification}
    Let $\C$ be an Induced Centrosymmetric Alternant code of length $n\geq 3$ induced by a Cauchy code on $\F_q$, with $F=\gamma_{p^j}\circ\overline f$ and $j\in\{0,\frac m2\}$. Assume that $F$ has one of the following forms. Then the corresponding conditions hold.
    \begin{enumerate}
        \item $F(x) = x^{p^{j}} $,and for $0\leq i\leq n-1$, $x_{\sigma^{-1}(i)}=x_i^{p^{j}}$ and $y_{\sigma^{-1}(i)}=\lambda y_i^{p^{j}}$ and order of $\lambda$ divide  $p^{j}+1$ and $j = \frac{m}{2}$

		\item $F(x) = ax^{p^{j}} $, $a$ non-zero and for $0\leq i\leq n-1$, $x_{\sigma^{-1}(i)}=a x_i^{p^{j}}$ and $y_{\sigma^{-1}(i)}=\lambda y_i^{p^{j}}$
		and the order $a,\lambda$ divide $p^{j}+1$.

        If \(j=0\), the condition \(a^2=1\) admits a nontrivial solution only when \(p>2\). Hence, for $j=0$ we must assume that \(p>2\).
		\item $F(x)=(x+b)^{p^{j}}$, $b \neq 0$ and for $0\leq i\leq n-1$, $x_{\sigma^{-1}(i)}=(x_i+b)^{p^{j}}$ and $y_{\sigma^{-1}(i)}=\lambda y_i^{p^{j}}$
		and $\lambda$ of order dividing $p^{j}+1$ and $b^{p^{j}} + b = 0$.

        For \(j=0\), we have \(2b=0\), since \(b\neq 0\), this implies that \(p=2\).
        \item $F(x)=(\frac{1}{cx+d})^{p^{j}}$, $c \neq 0$ and for $0\leq i\leq n-1$, $x_{\sigma^{-1}(i)}=(\frac{1}{cx_i + d})^{p^{j}}$ and $y_{\sigma^{-1}(i)}=(\lambda x_i^{k-1} y_i)^{p^{j}}$ with $\lambda^{p^{j}+1} = c^{k-1}$ satisfying :
        \begin{itemize}
            \item $d = 0$ for $j = 0$
            \item $d = 0$ and $c^{p^{\frac{m}{2}}-1} = 1$ for $j = \frac{m}{2}$
        \end{itemize}
    \end{enumerate}
\end{theorem}
\begin{proof}
    Thanks to Lemma~\ref{lem:Nid}, $F$ must satisfy $F(F(x)) = x$.
    We now list, for each case, the conditions under which $C$ induces a centrosymmetric alternant code with support $x$ and  multiplier $y$ defined by \\
   $x_{\sigma^{-1}(i)} = F(x_{i})$ and $y_{\sigma^{-1}(i)}=(\lambda^{'} u_f(x_i)^{k-1}y_i)^{p^j}$.\\
    In the following we suppose that $y_{\sigma^{-1}(i)}= \lambda y_i^{p^j}$ with $\lambda = (\lambda^{'} u_f(x_i)^{k-1})^{p^j}$ for the case $1,2$ and $3$.
    \begin{enumerate}
       \item $F(x)=x^{p^j}$.

    The conditions on the support imply $x_i^{p^{2j}}=x_i$, while the conditions relating to the multiplier yield $\lambda^{p^j+1}y_i^{p^{2j}}=y_i$.
    Since the permutation $\sigma$ is non trivial, the case $j=0$ is excluded, so $j=m/2$. We thus obtain $\lambda^{p^j+1}=1$,
    which means that the order of $\lambda$ divides $p^j+1$.
       \item $F(x) = ax^{p^j}$ \\
       Since $F(F(x_i))=x_i$, we obtain $a(a x_i^{p^j})^{p^j}=x_i$ and $\lambda(\lambda y_i^{p^j})^{p^j}=y_i$.
        Equivalently, $a^{p^j+1}x_i^{p^{2j}}=x_i$ and $\lambda^{p^j+1}y_i^{p^{2j}}=y_i$.
        These relations imply that $a^{p^j+1}=1$ and $\lambda^{p^j+1}=1$ for $j=0$ or $j = \frac m2$.

        In the case $j=0$, the condition $a^2=1$ requires the
        characteristic to be odd if $a\ne 1$. Indeed, in characteristic
        two, the only solution to $a^2=1$ is $a=1$, which does not allow for
        the non-trivial transformation considered here.
        \item $F(x)=(x+b)^{p^j}$, where $b\ne 0$.

Since $F(F(x_i))=x_i$, we obtain $\bigl((x_i+b)^{p^j}+b\bigr)^{p^j}=x_i$,
and $\lambda(\lambda y_i^{p^j})^{p^j}=y_i$.
Equivalently, $(x_i^{p^j}+b^{p^j}+b)^{p^j}=x_i$,
and $\lambda^{p^j+1}y_i^{p^{2j}}=y_i$.
These relations imply that $\lambda^{p^j+1}=1$ and $b+b^{p^j}=0$ where $j=0$ or $j=m/2$

Consequently, the multiplicative order of $\lambda$ divides $p^j+1$.

Since $b\ne 0$, the condition $b+b^{p^j}=0$ can only be satisfied for
$j=0$ if the characteristic is equal to two. Thus, this case requires
the characteristic to be two when $j=0$.
\item $F(x)=\left(\frac{1}{cx+d}\right)^{p^j}$, where $c\ne 0$.

Since $F(F(x_i))=x_i$, a direct calculation yields $\frac{c^{p^{2j}}x_i^{p^{2j}}+d^{p^{2j}}}
{c^{p^j}+d^{p^j}c^{p^{2j}}x_i^{p^{2j}}
+d^{p^{2j}+p^j}}=x_i.$
\begin{itemize}
\item If $j=0$, the equality reduces to $cx_i+d=cx_i+dcx_i^2+d^2x_i \implies d(c x_i^2 + dx_i - 1) = 0$. Since $n \geq 3$ and $c\ne 0$, the polynomial $cX^2+dX-1$ cannot vanish on the whole support, so $d = 0$.
\item If $j=m/2$, the equality becomes $cx_i+d=c^{p^{m/2}}x_i
+d^{p^{m/2}}cx_i^2+d^{p^{\frac m2}+1}x_i \implies d^{p^{m/2}}cx_i^2 + (c^{p^{m/2}} + d^{p^{\frac m2}+1} - c) x_i - d = 0$ since $n \geq 3$  and $c\ne 0$, we have $d=0$ and $c= c^{p^{m/2}}$
\end{itemize}
Now consider the multiplier condition $y_{\sigma^{-1}(i)}
=\left(\lambda x_i^{k-1}y_i\right)^{p^j}$.
Since $x_{\sigma^{-1}(i)}=\left(\frac{1}{cx_i}\right)^{p^j}$,
applying the multiplier condition twice yields $\lambda^{p^j+p^{2j}}c^{-(k-1)p^{2j}} y_i^{p^{2j}}=y_i$. In the case $j=0$ this gives $\lambda^{2}=c^{k-1}$. In the case $j=m/2$, where $y_i^{p^{2j}}=y_i$ and $c^{p^{2j}}=c$, this implies $\lambda^{p^j+1}= c^{k-1}$ .
     \end{enumerate}
\end{proof}


\begin{theorem}\label{thm:normal}
Let $A_k(x,y)$ be an induced centrosymmetric alternant code of length $n\geq 3$ with $j\in\{0,\frac m2\}$. Then there exist a support $u$ and a multiplier $v$ such that $GRS_k(x;y)=GRS_k(u;v)$ and one of the following holds for $0\le i<n$.
\begin{enumerate}
\item $j=\frac m2$: $u_{n-1-i}=u_i^{p^{m/2}}$ and $v_{n-1-i}=v_i^{p^{m/2}}$.
\item $j=0$ and $p$ odd: $u_{n-1-i}=-u_i$ and $v_{n-1-i}=\lambda v_i$ with $\lambda=\pm 1$.
\item $j=0$, $p$ odd, $k$ odd and $n$ even: $u_{n-1-i}=\varepsilon/u_i$ and $v_{n-1-i}=\lambda u_i^{k-1}v_i$ with $\lambda^2=\varepsilon^{-(k-1)}$, where $\varepsilon$ is a fixed non-square.
\item $j=0$ and $p=2$: $u_{n-1-i}=u_i+1$ and $v_{n-1-i}=v_i$.
\end{enumerate}
Conversely each of these conditions gives an induced centrosymmetric alternant code.
\end{theorem}
\begin{proof}
By Lemma~\ref{lem:Nid} we have $N(f)=id$. By Corollary~\ref{cor:involutions}, $f$ is $\gamma_{p^j}$-similar to a matrix $g$ which is the identity if $j=\frac m2$, and one of
$\left(\begin{array}{cc} -1 & 0 \\ 0 & 1 \end{array}\right)$,
$\left(\begin{array}{cc} 0 & \varepsilon \\ 1 & 0 \end{array}\right)$,
$\left(\begin{array}{cc} 1 & 1 \\ 0 & 1 \end{array}\right)$ if $j=0$. For these matrices $u_g(z)$ is equal to $1$, $1$, $z$, $1$ respectively. By Theorem~\ref{lemmegamma} there exist $u$ and $v$ with $GRS_k(x;y)=GRS_k(u;v)$,
$u_{n-1-i}=\gamma_{p^j}\circ\overline g(u_i)$ and $v_{n-1-i}=(\lambda u_g(u_i)^{k-1}v_i)^{p^j}$. We apply the last relation twice.
\begin{itemize}
\item If $j=\frac m2$ and $g=id$, we get $v_i=\lambda^{r+1}v_i$ with $r=p^{m/2}$, so $N_r(\lambda)=1$. If we replace $v$ by $\mu v$, then $\lambda^{r}$ is replaced by $\lambda^{r}\mu^{1-r}$. By Lemma~\ref{lem:h90} we can choose $\mu$ such that $\mu^{r-1}=\lambda^{r}$, and then $\lambda=1$.
\item If $j=0$ and $\overline g(z)=-z$ or $\overline g(z)=z+1$, we get $v_i=\lambda^2v_i$, so $\lambda=\pm1$, and $\lambda=1$ if $p=2$.
\item If $j=0$ and $\overline g(z)=\varepsilon/z$, we get $v_i=\lambda(\varepsilon/u_i)^{k-1}\lambda u_i^{k-1}v_i=\lambda^2\varepsilon^{k-1}v_i$, so $\lambda^2=\varepsilon^{-(k-1)}$. Since $\varepsilon$ is not a square, this is possible only if $k$ is odd. Moreover $\overline g$ has no fixed point, so $n$ is even.
\end{itemize}
The converse is given by Theorem~\ref{tho_3}.
\end{proof}

A parity condition of the same kind as in the third case appears in \cite[Theorem 2.10]{Berger1999}, for cyclic codes induced by a homography without fixed point.
\begin{remark}
The support $u$ of Theorem~\ref{thm:normal} may contain the point $\infty$. The matrix $h$ used in the change of support can be multiplied by any matrix commuting with $g$, and this moves the point sent to $\infty$ among the fixed points of $F$ in the cases $1$, $2$, $4$, and anywhere in the case $3$. A point $x_i$ with $n-1-i\neq i$ is not fixed by $F$. So when $n$ is even we can always take $u\in\F_q^n$.
\end{remark}
\begin{remark}
Theorem~\ref{alternant_classification} 
\begin{enumerate}
\item When \(j=0\), we obtain three types of transformations in Theorem~\ref{alternant_classification}: \(F(x)=ax\), with \(a^2=1\) and \(p>2\); \(F(x)=x+b\), with \(b\ne 0\) and \(p=2\); and \(F(x)=1/(cx)\), with \(c\ne 0\). They do not give three different classes. If $p=2$, the map $x\mapsto 1/(cx)$ has only one fixed point and it is conjugate to $x\mapsto x+1$. If $p$ is odd, it is conjugate to $x\mapsto -x$ when $c$ is a square, and it has no fixed point when $c$ is not a square. In this last case the condition $\lambda^2=c^{k-1}$ requires $k$ to be odd.

\item When \(j=m/2\), all the transformations considered in Theorem~\ref{alternant_classification} have trivial Shintani norm. Indeed, for the four families, the conditions obtained in Theorem~\ref{alternant_classification} are exactly $N(f)=f^{(p^{m/2})}\circ f=\mathrm{id}$: they read $N_r(a)=1$, $T_r(b)=0$ and $c\in\F_r^{*}$.

By Corollary~\ref{cor:involutions}, they therefore belong to the same \(\gamma_{p^{m/2}}\)-similarity class as the identity. Equivalently, their associated semilinear transformations are conjugate to the Frobenius automorphism. Nevertheless, their distinct explicit forms remain useful for constructing supports and multipliers satisfying the induced-centrosymmetry conditions. By Theorem~\ref{lemmegamma} the four families give the same set of codes, namely the codes of the first case of Theorem~\ref{thm:normal}.

\item When \(j=m/2\) and $n$ is even, the support does not contain the $p^{m/2}+1$ fixed points of $F$. Hence $n\le q-\sqrt q$.
\end{enumerate}
\end{remark}
\section{Examples}\label{exemple}
We illustrate Theorem~\ref{alternant_classification} over the finite field $K=\mathbb{F}_{2^6}=\mathbb{F}_{64}$ and let $\alpha$ be a primitive element of $\mathbb{F}_{2^{6}}$, root of $x^6 + x^4 + x^3 + x + 1$.

Throughout this section, we consider a GRS code of length $n=16$
and dimension $k=2$, together with the centrosymmetric permutation
\[
\sigma(i)=15-i,\qquad 0\le i\le 15.
\] and $J_m$ denotes the $m\times m$ anti-identity matrix.
Let $\lambda=\alpha^5+\alpha^4+\alpha^2+\alpha$ be an element of order $9$.

For the first three cases, we use the same multiplier vector
$v=(v_0,\ldots,v_{15})\in K^{16}$, given by
\[
\begin{aligned}
    v = \big(&\alpha^{3} + \alpha^{2} + 1,\,\alpha^{5} + \alpha^{3} + \alpha^{2} + \alpha + 1,\,\alpha^{5} + \alpha^{4} + \alpha,\,\alpha^{5} + \alpha^{4} + \alpha^{3} + \alpha^{2} + \alpha,\\
    &\alpha^{4} + \alpha^{3} + \alpha^{2},\,\alpha^{2} + \alpha,\,\alpha^{4} + \alpha^{3} + \alpha^{2} + \alpha,\,\alpha^{5} + \alpha^{4} + \alpha^{3} + \alpha + 1,\,\alpha^{4} + \alpha^{2} + \alpha,\\
    &\alpha^{5} + \alpha,\,\alpha^{4} + \alpha^{3} + \alpha^{2},\,\alpha^{2} + \alpha,\,\alpha^{4} + \alpha^{3},\,\alpha^{5} + \alpha^{4} + \alpha^{3} + 1,\,\alpha^{3} + 1,\,\alpha^{4} + \alpha^{2} + \alpha + 1 \big)
\end{aligned}
\]

We now consider the four types of transformations appearing in
Theorem~\ref{alternant_classification}. For each case, we specify
the corresponding support and give a \emph{centrosymmetric} generator matrix
$G_c$ of the associated binary alternant code $C$, $J_{4}\, G_c\, J_{16} = G_c$.
In each case the alternant code has dimension $4$: the general bound gives $\dim C\geq n-mk=16-6\cdot 2=4$, and the equality is obtained by computation. By Corollary~\ref{cor:involutions}, the four cases belong to the same $\gamma_{p^{m/2}}$-similarity class.
\medskip

\noindent\textbf{(i) The case $f(x)=x^{8}$.}

Consider the support $L^{(1)}=(x_0,\ldots,x_{15})$, where
\[
\begin{aligned}
    L^{(1)} = \big( &\alpha, \alpha^{2}, \alpha^{3}, \alpha^{4}, \alpha^{5}, \alpha^{4} + \alpha^{3} + \alpha + 1, \alpha^{5} + \alpha^{4} + \alpha^{2} + \alpha, \alpha^{5} + \alpha^{4} + 1,\alpha^{5} + \alpha^{2} + \alpha\\
    &, \alpha + 1, \alpha^{5} + \alpha^{3} + \alpha^{2} + \alpha + 1, \alpha^{4} + \alpha^{2} + 1, \alpha^{5} + \alpha^{4} + \alpha, \alpha^{5} + \alpha^{3} + \alpha + 1,\\
    &\alpha^{4} + \alpha + 1, \alpha^{5} + \alpha^{4} + \alpha^{2} + \alpha + 1 \big)
\end{aligned}
\]
which satisfies
$
x_{\sigma(i)}=x_i^{8},
\qquad 0\le i\le 15.
$

Let \(C=A_2(L^{(1)},v)\) be the alternant code, that is the subfield subcode of the dual of the generalized Reed–Solomon code \(GRS_2(L^{(1)},v)\). It has length \(n=16\) and dimension \(4\). A centrosymmetric generator matrix of \(C\) is
\[
G_c = \left(\begin{array}{rrrrrrrrrrrrrrrr}
1 & 0 & 0 & 0 & 0 & 1 & 0 & 1 & 0 & 0 & 0 & 0 & 0 & 1 & 1 & 0 \\
0 & 0 & 0 & 1 & 1 & 0 & 0 & 1 & 0 & 1 & 1 & 1 & 1 & 1 & 0 & 0 \\
0 & 0 & 1 & 1 & 1 & 1 & 1 & 0 & 1 & 0 & 0 & 1 & 1 & 0 & 0 & 0 \\
0 & 1 & 1 & 0 & 0 & 0 & 0 & 0 & 1 & 0 & 1 & 0 & 0 & 0 & 0 & 1
\end{array}\right)
\] 
Therefore, the constructed alternant code is invariant under the centrosymmetric permutation induced by the transformation \(x\mapsto x^8\).
\medskip

\noindent\textbf{(ii) The case $f(x)=a x^{8}$, $a = \alpha^4 + \alpha$}

Using the same multiplier vector $v$, consider the support
$L^{(2)}=(x_0,\ldots,x_{15})$,
\[
\begin{aligned}
    L^{(2)} = \big( &\alpha, \alpha^{2}, \alpha^{3}, \alpha^{4}, \alpha^{5}, \alpha^{4} + \alpha^{3} + \alpha + 1, \alpha^{5} + \alpha^{4} + \alpha^{2} + \alpha + 1, \\
    &\alpha^{5} + \alpha^{4} + \alpha^{2} + 1, \alpha^{5} + \alpha^{4} + \alpha^{3}, \alpha^{5} + \alpha^{2}, \alpha^{5} + \alpha^{3} + \alpha^{2} + 1, \alpha^{4} + \alpha^{2} + \alpha,\\
    &\alpha^{4} + \alpha^{3} + \alpha^{2} + 1, \alpha^{5} + \alpha^{4} + \alpha^{3} + \alpha^{2} + \alpha, \alpha^{5} + 1, \alpha^{4} + \alpha^{3} + \alpha^{2}\big)
\end{aligned}
\]
satisfying
\[
x_{\sigma(i)}= a x_i^8,\qquad 0\le i\le 15.
\]
Let \(C=A_2(L^{(2)},v)\) be the alternant code of length \(n=16\) and dimension \(4\). A centrosymmetric generator matrix of \(C\) is
\[
G_c = \left(\begin{array}{rrrrrrrrrrrrrrrr}
0 & 0 & 0 & 1 & 1 & 1 & 0 & 1 & 1 & 0 & 0 & 0 & 0 & 0 & 1 & 0 \\
0 & 0 & 1 & 0 & 1 & 0 & 1 & 1 & 1 & 1 & 0 & 0 & 1 & 1 & 1 & 1 \\
1 & 1 & 1 & 1 & 0 & 0 & 1 & 1 & 1 & 1 & 0 & 1 & 0 & 1 & 0 & 0 \\
0 & 1 & 0 & 0 & 0 & 0 & 0 & 1 & 1 & 0 & 1 & 1 & 1 & 0 & 0 & 0
\end{array}\right)
\]
Therefore, the alternant code is invariant under the
centrosymmetric permutation induced by $x\mapsto a x^8$.

\medskip
\noindent\textbf{(iii) The case $f(x)=(x+b)^{8}$.}
Let $b=\alpha^5 + \alpha^4 + \alpha^2 + 1$, and consider the support
$L^{(3)}=(x_0,\ldots,x_{15})$, satisfying
\[
x_{\sigma(i)}=(x_i+b)^8,\qquad 0\le i\le 15.
\]
Again, we use the same multiplier vector $v$ as in Cases 1 and 2.
\[
\begin{aligned}
    L^{(3)} = \big(&0, \alpha^{2}, \alpha^{3}, \alpha^{4}, \alpha^{4} + \alpha^{3} + \alpha + 1, \alpha^{5} + \alpha^{4} + 1, \alpha^{5} + \alpha^{4} + \alpha^{3} + 1, \alpha^{5} + \alpha^{3} + 1,\\
    &\alpha^{3} + \alpha, \alpha^{5} + \alpha^{4} + \alpha^{3}, \alpha^{4} + \alpha + 1, \alpha^{4} + \alpha^{3} + \alpha, \alpha^{2} + \alpha + 1, \\
    &\alpha^{4} + \alpha^{3} + \alpha^{2} + \alpha, \alpha^{5} + \alpha^{2} + \alpha, \alpha^{5} + \alpha^{4} + \alpha^{2} + 1 \big)
\end{aligned}
\]
Let \(C=A_2(L^{(3)},v)\) be the alternant code of length \(n=16\) and dimension \(4\) , and a centrosymmetric generator matrix of $C$ is
\[
G_c = \left(\begin{array}{rrrrrrrrrrrrrrrr}
1 & 0 & 1 & 0 & 1 & 0 & 0 & 1 & 0 & 0 & 0 & 0 & 1 & 0 & 0 & 0 \\
1 & 0 & 1 & 0 & 0 & 0 & 0 & 0 & 0 & 0 & 0 & 1 & 0 & 1 & 0 & 0 \\
0 & 0 & 1 & 0 & 1 & 0 & 0 & 0 & 0 & 0 & 0 & 0 & 0 & 1 & 0 & 1 \\
0 & 0 & 0 & 1 & 0 & 0 & 0 & 0 & 1 & 0 & 0 & 1 & 0 & 1 & 0 & 1
\end{array}\right)
\]

Therefore, the constructed alternant code is invariant under the centrosymmetric permutation induced by the transformation \(x\mapsto (x+b)^8\).

\medskip

\noindent\textbf{(iv) The case $f(x)=(\frac{1}{cx})^{8}$.}

For $c=1$, consider the support $L^{(4)}=(x_0,\ldots,x_{15})$,
satisfying
\[
x_{\sigma(i)}=x_i^{-8},\qquad 0\le i\le 15.
\]
\[
\begin{aligned}
    L^{(4)} = \big(& \alpha, \alpha^{2}, \alpha^{3}, \alpha^{4}, \alpha^{5}, \alpha^{4} + \alpha^{3} + \alpha + 1, \alpha^{5} + \alpha^{4} + \alpha^{2} + \alpha + 1, \alpha^{5} + \alpha^{4} + \alpha^{2} + 1,\\
    &\alpha^{4} + \alpha^{2} + \alpha, \alpha^{5} + \alpha^{3} + \alpha^{2} + 1, \alpha^{5} + \alpha^{2}, \alpha^{5} + \alpha^{4} + \alpha^{3}, \alpha^{4} + \alpha^{3} + 1,\\
    &\alpha^{5} + \alpha^{2} + \alpha + 1, \alpha^{5} + \alpha^{4} + \alpha^{3} + \alpha, \alpha^{5} + \alpha^{3} + \alpha^{2} \big)
\end{aligned}
\]
In this case, we use a multiplier vector $w$ satisfying \[
w_{\sigma(i)}=(\lambda x_i w_i)^8,\qquad 0\le i\le 15.
\]
\[
\begin{aligned}
    w = \big(& \alpha^{5} + \alpha^{2},\,\alpha,\,\alpha^{5} + \alpha^{4} + \alpha^{3} + \alpha^{2} + 1,\,\alpha^{5} + \alpha^{3} + \alpha,\,\alpha^{4} + \alpha^{3} + 1, \\
    &\alpha^{5} + \alpha^{2},\,\alpha^{3},\,\alpha^{5} + \alpha^{4} + \alpha^{3} + \alpha^{2} + 1,\,\alpha^{5} + \alpha + 1,\,\alpha^{4} + \alpha^{2} + \alpha + 1,\,\alpha^{4} + 1,\\
    &\alpha^{5} + \alpha^{4} + \alpha^{3} + \alpha^{2} + 1,\,\alpha^{5} + \alpha^{4} + \alpha^{3} + \alpha^{2} + \alpha,\,\alpha^{5} + \alpha^{4},\,\alpha^{5} + \alpha^{2} + \alpha,\,\alpha^{3} + \alpha^{2} \big)
\end{aligned}
\]
Let \(C=A_2(L^{(4)},w)\) be the alternant code of length \(n=16\) and dimension \(4\). A centrosymmetric generator matrix of \(C\) is
\[
G_c = \left(\begin{array}{rrrrrrrrrrrrrrrr}
1 & 0 & 1 & 0 & 1 & 1 & 1 & 1 & 0 & 0 & 1 & 0 & 0 & 0 & 0 & 0 \\
0 & 0 & 0 & 1 & 1 & 0 & 0 & 1 & 0 & 1 & 0 & 1 & 1 & 1 & 1 & 1 \\
1 & 1 & 1 & 1 & 1 & 0 & 1 & 0 & 1 & 0 & 0 & 1 & 1 & 0 & 0 & 0 \\
0 & 0 & 0 & 0 & 0 & 1 & 0 & 0 & 1 & 1 & 1 & 1 & 0 & 1 & 0 & 1
\end{array}\right)
\]

Therefore, the constructed alternant code is invariant under the centrosymmetric permutation induced by the transformation \(x\mapsto x^{-8}\).
\section{Conclusion}\label{sec:conclusion}
In this paper, we have explored the construction of induced centrosymmetric alternant codes through the automorphism structures of generalized Reed--Solomon codes. The introduction of twisted conjugation provides a natural way to study semilinear projective transformations and connects their classification to ordinary conjugacy in the corresponding semilinear coset.

By using the Shintani norm, we identified four structural types of transformations and showed that, in the quadratic extension case, all admissible semilinear involutions have trivial Shintani norm and belong to a unique conjugacy class. This classification allowed us to derive explicit constructions of induced centrosymmetric alternant codes and to bring previously known constructions, including quasi-centrosymmetric alternant codes, into a common framework.

An interesting direction for future work is to extend this approach to induced centrosymmetric Goppa codes and further investigate the algebraic properties of these code families.
 
\newpage
\bibliographystyle{plain}
\bibliography{code}
\renewcommand{\arraystretch}{1}
\end{document}